\documentclass[11pt]{article}

\usepackage[utf8]{inputenc}
\usepackage{graphicx}
\usepackage{hyperref}
\usepackage{amsfonts,amsmath,amsthm}
\usepackage{enumitem}
\usepackage{relsize}
\usepackage{bm}
\usepackage{caption}
\usepackage{fullpage}
\usepackage[ruled, algo2e]{algorithm2e} % For algorithms

\SetAlFnt{\small}
\SetAlCapFnt{\small}
\SetAlCapNameFnt{\small}
\SetAlCapHSkip{0pt}
\IncMargin{-\parindent}

\usepackage{amsopn}
\DeclareMathOperator{\diag}{diag}

\usepackage{xspace}
\newcommand{\algname}{Zoltan-AlgD\xspace}
\newcommand{\real}{\mathbb{R}}

\newcommand{\patohqcolor}{aquamarine\xspace}
\newcommand{\patohdcolor}{blue\xspace}
\newcommand{\hmetiscolor}{orange\xspace}
\newcommand{\zoltancolor}{brown\xspace}

\newtheorem{theorem}{Theorem}

\usepackage{cleveref}

\author{
  Ruslan Shaydulin \\ 
  School of Computing\\ 
  Clemson University \\
    rshaydu@clemson.edu
  \and
  Jie Chen \\ 
  IBM Thomas J. Watson\\ Research Center \\
  chenjie@us.ibm.com
  \and
  Ilya Safro \\ 
  School of Computing \\ 
  Clemson University\\
  isafro@clemson.edu
}

\title{Relaxation-Based Coarsening for Multilevel Hypergraph Partitioning}

\begin{document}

\maketitle

\begin{abstract}
Multilevel partitioning methods that are inspired by principles of multiscaling are the most powerful practical hypergraph partitioning solvers. 
Hypergraph partitioning has many applications in disciplines ranging from scientific computing to data science. In this paper we introduce the concept of algebraic distance on hypergraphs and demonstrate its use as an algorithmic component in the coarsening stage of multilevel hypergraph partitioning solvers. The algebraic distance is a vertex distance measure that extends hyperedge weights for capturing the local connectivity of vertices which is critical for hypergraph coarsening schemes.  
The practical effectiveness of the proposed measure and corresponding coarsening scheme is demonstrated through extensive computational experiments on a diverse set of problems. Finally, we propose a benchmark of hypergraph partitioning problems to compare the quality of other solvers.
\end{abstract}

\section{Introduction}

Hypergraphs are generalizations of graphs. Both graphs and hypergraphs are ordered pairs of sets $(V, E)$, where $V$ is the vertex set and $E$ is the set of (hyper)edges such that each $e\in E$ is a subset of vertices. The difference is that in a graph, the cardinality of each edge is exactly two, whereas in a hypergraph, a hyperedge can contain an arbitrary number of vertices. 
The hypergraph partitioning (HGP) problem is therefore a generalization of the graph partitioning (GP) problem. In GP the goal is to split the set of vertices into multiple sets (usually called partitions) of similar sizes, such that a cut metric is minimized. Here, a cut defines a set of (hyper)edges spanning more than one partition. In the HGP generalization, the hyperedges can possibly span more than two partitions. There exist several versions of minimization objectives, constraints, and cut metrics in both GP and HGP~\cite{bichot2013graph,bulucc2016recent}.
Hypergraph partitioning has many applications, including VLSI design~\cite{karypis1999multilevel, alpert1995recent}, parallel matrix multiplication~\cite{catalyurek1999hypergraph}, classification~\cite{zhou2006learning}, cluster ensembling~\cite{strehl2002cluster}, and combinatorial scientific computing \cite{Naumann:2012:CSC:2141107}, among others~\cite{bichot2013graph,papa2007hypergraph}.

\subsection{Multilevel partitioning}

One of the most popular and fastest approaches for solving HGP problems is multilevel: first, the hypergraph is iteratively coarsened by  merging its vertices; then, the solution (i.e., the partitioning) is computed at the coarsest level; and finally, the coarsest solution is gradually uncoarsened back to the finest level (see \cref{fig:mainscheme}). Coarsening consists of two steps, namely, aggregation and contraction. During the aggregation step, the decision is made on which vertices to merge. At the contraction step, these vertices are merged and hyperedges are dropped. The reverse process, uncoarsening, also consists of two steps, namely, interpolation and refinement. In the interpolation step, the solution for coarse vertices (i.e., the assignment of vertices into corresponding partitions) is interpolated into the next-fine level. The refinement step improves the interpolated solution, typically, by iteratively solving a sequence of local optimization problems that refine the objective or constraints of HGP. In many state-of-the-art solvers, such refinement is implemented using vertex-moving heuristics that accommodate a vertex in a better partition.

The multilevel method captures the global structure of a hypergraph in the solution by combining local information at different scales of coarseness. The aims of coarsening are (a) preserving the structural properties of the original hypergraph throughout the entire multilevel hierarchy until the smallest hypergraph is obtained, and (b) allowing an effective interpolation in which the solution of coarse level variables is correctly projected into the corresponding fine level ones. Typically, a high-quality coarsening should also preserve the spectral properties of the (hyper)graph adjacency matrix (even though it is not directly required by a coarsening algorithm).  The coarsest hypergraph is partitioned, and the solution is projected on larger, finer-level hypergraphs all the way back to the original one. At each level the projected solution is refined.

\begin{figure}[h]
  \includegraphics[width=\textwidth]{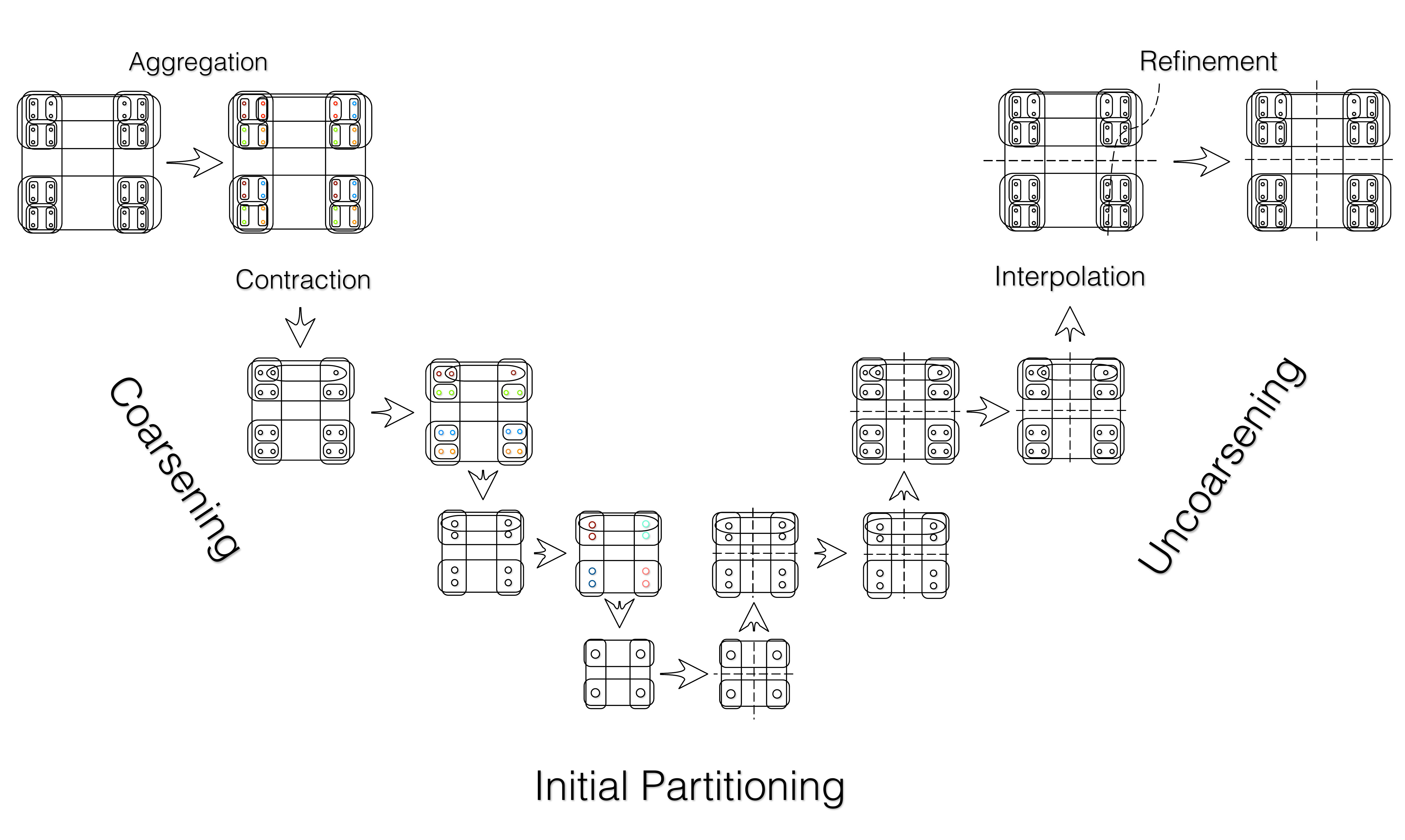}
  \caption{Outline of multilevel hypergraph partitioning scheme.  On the left side of the V-cycle (named after the shape of the coarsening-uncoarsening hierarchy), the hypergraph is iteratively coarsened by aggregating (horizontal arrows) and contracting (vertical arrows) vertices. On the right side of the V-cycle, the initial partitioning is repeatedly interpolated (vertical arrows) into the next-fine level and refined (horizontal arrows).}\label{fig:mainscheme}
\end{figure}

As a result, the multilevel method ends up with a solution that is aware of the global structure and is locally optimal as well. Note, however, that the global optimality of the final solution heavily depends on the quality of coarsening. Therefore, it is paramount to have a good distance/similarity measure on vertices, one that captures not only how vertices are similar \emph{locally}, but also how close they are \emph{globally}. In other words, a sophisticated similarity measure for vertices (and a resulting more accurate aggregation) allows for a more accurate interpolation of the solution in the next step.

In most cases, multilevel algorithms for HGP are generalizations of those for GP and other cut optimization problems; they are, in turn, simplified adaptations of multigrid and other multiscale optimization strategies \cite{brandt:optstrat}. There exist a few coarsening methods for different cut problems on graphs; e.g., various versions of weighted aggregations inspired by algebraic multigrid for graph reordering \cite{Safro2006,SafroT11,RonSB11,Hu:wavefront}, and matching-based methods for GP \cite{karypis1999multilevel,Sanders2011}, vertex separator problem \cite{safro:mlvsp}, and clustering \cite{DhillonGK07weighted}. For GP, advanced similarity measures used in coarsening include flows \cite{delling2011graph}, diffusion \cite{lafon2006diffusion}, and algebraic distances \cite{amg-sss12}. For an in-depth overview, the reader is referred to a survey by Bulu\c{c} et al.~\cite{bulucc2016recent}. However, to the best of our knowledge there exist few advanced similarity concepts for coarsening applied to HGP.

\subsection{Our contribution}

We introduce a similarity measure for vertices of hypergraphs, as well as a relaxation-based coarsening scheme for multilevel hypergraph partitioning.
The measure is called \emph{algebraic distance for hypergraphs}. It generalizes the algebraic distance developed for  graphs \cite{chen2011algebraic,RonSB11}, used in applications such as graph partitioning \cite{amg-sss12}, sparsification \cite{john2016single}, and other problems \cite{SafroT11,safro-fastresp}. The algebraic distance takes into account distant neighborhoods of vertices and addresses the issue of unweighted and noisy hyperedges.  At each level of coarsening, we compute algebraic distances on a bipartite representation  of the hypergraph and modify the hyperedge weights
that are used for agglomerative inner-product matching~\cite{catalyurek2009repartitioning}.
This way, we penalize cutting hyperedges that contain similar vertices measured by the algebraic distances between them.  
Our implementation is available at \url{https://github.com/rsln-s/algebraic-distance-on-hypergraphs}. We also propose a comprehensive benchmark of hypergraphs and the partitioning results for a broad scientific community to evaluate the quality of their solvers. Our benchmark is available at  \url{http://shaydul.in/hypergraph-partitioning-archive/}. 

\section{Preliminaries}

A hypergraph is an ordered pair of sets $(V,E)$, where $V$ is the set of vertices and $E$ is the set of hyperedges. Each hyperedge is a subset of $V$. Another terminology is also used in the literature, wherein a hypergraph is defined as a triplet $(V,H,P)$, where $V$ is the set of vertices, $H$ is the set of hyperedge labels and $P$ is the set of pins. Each pin connects a vertex from $V$ with a label from $H$. A net (a synonym of hyperedge) is therefore a set of pins connecting the vertices to the corresponding label~\cite{bichot2013graph}. In this paper, we will use only the former terminology. We will also refer to a hyperedge as simply an ``edge'' when it does not cause confusion.
Both vertices and hyperedges can have weights, namely, $w(v)\in\mathbb R_{\geq 0}$ and $w(e)\in\mathbb R_{\geq 0}$ respectively for each $v\in V, e \in E$. A zero weight for a hyperedge assumes that it is not present in the hypergraph.

\subsection{Hypergraph partitioning}

Hypergraph partitioning is a generalization of graph partitioning. The aim is to partition the set of vertices into a number of disjoint sets such that a cut metric is minimized subject to some imbalance constraint. Partitioning into two parts is often referred to as \emph{bipartitioning}, and partitioning into $k$ parts, \emph{$k$-way partitioning}.

More formally, if $V$ is the set of vertices and $E$ is the set of edges, a \emph{partitioning} is a set of $k$ mutually disjoint subsets $V_i \subseteq V$, $1\leq i \leq k$, such that $V_1\cup V_2\cup .. \cup V_k = V$. Then, the cut is defined as $E_{cut}\subset E$ such that for $e\in E$ holds $e\in E_{cut}$ if and only if $\exists i, j\in \mathbb{N}$, $i\neq j$: $e\cap V_i \neq \emptyset, e\cap V_j\neq\emptyset$.

Several cut metrics have been introduced over the years~\cite{Karypis97multilevelhypergraph,alpert1995recent,catalyurek1999hypergraph} for HGP. Two of the most commonly used are \emph{hyperedge cut} and \emph{connectivity}.
The hyperedge cut cost is equal to the total weights of edges that span more than one partition, i.e.,
\[
\sum_{e\in E_{cut}}w(e).
\]
The connectivity objective is defined as
\[
\sum_{e\in E_{cut}}(\lambda(e)-1),
\]
where $\lambda(e)$ is the number of partitions spanned by the hyperedge $e$~\cite{devine2006parallel}. In this paper we aim at minimizing the hyperedge cut metric.

The imbalance constraint specifies to what extent the partitions can differ in their size. In this paper we will use the imbalance constraint defined in the state-of-the-art hypergraph partitioner Zoltan~\cite{devine2006parallel}. Imbalance is defined as the sum of vertex weights in the maximum partition over the average sum of vertex weights in a partition. For $k$-way partitioning:

 \[
 imbal=\dfrac{\sum_{v\in V_{max}}w(v)}{\frac{1}{k}\sum_{v\in V}w(v)},
 \]
 where $V_{max}$ is the largest partition (i.e. $ \sum_{v\in V_{max}}w(v) = \max_i\left(\sum_{v\in V_{i}}w(v)\right)$).

\subsection{Multilevel method}
\label{sec:prereq_multilevel}

The multilevel method for graph partitioning was initially introduced to speed up existing algorithms~\cite{barnard1994fast}, but was quickly recognized as a good way to improve the quality of partitioning~\cite{karypis1998fast,hendrickson1995multi}. It is used as a global suboptimal heuristic framework~\cite{bulucc2016recent}, in which other heuristics are incorporated at different stages. These three stages are coarsening, initial partitioning, and uncoarsening.

During the coarsening stage a hypergraph $H = (V,E)$ is approximated via a series of successively smaller hypergraphs $H^i = (V^i, E^i)$, $1\leq i \leq l$, where $l$ is the number of  levels in the hierarchy. The superscript denotes the number of the corresponding level for hypergraphs, nodes, and edges, respectively.
Each next-coarse hypergraph is constructed by merging or aggregating vertices in the previous one according to some heuristic: $v^i_k = \{v^{i-1}_{k_1}, ..., v^{i-1}_{k_j}\}$. That is, a vertex $v^i_k$ in the coarse hypergraph $G^i$ at the $i$th level is created by grouping a set of vertices $v^{i-1}_{k_1}, ..., v^{i-1}_{k_j}$ from the finer hypergraph $H^{i-1}$. Vertices can be grouped by using different criteria, with the aim of interpolating solution from coarse level nodes to the corresponding fine level nodes with minimum loss of solution quality. In the case of pairwise grouping of vertices, such a coarsening is referred to as matching. The weight of the new coarse vertex is equal to the sum of weights of the merged vertices: 
\[
w(v^i_k) = w(v^{i-1}_{k_1}) + ... + w(v^{i-1}_{k_j}).
\]
A coarse vertex is contained in all hyperedges that contain the merged vertices. Hyperedges of cardinality one are discarded. Coarsening terminates when the size of the hypergraph is below a certain threshold or when a solution is easy to compute. 

In the partitioning stage, the coarsest hypergraph $H^l$ is partitioned using exact or approximate solver. In many existing implementations, the solver is a local search heuristic. 
This partitioning is anticipated to approximate the \emph{global} solution in the sense that it incorporates the global structure of the hypergraph. In some cases, when $H^l$ is sufficiently small, an exact solution can be computed.

The uncoarsening stage consists of two steps, namely, interpolation and refinement. During uncoarsening, the partitioning from the coarse hypergraph $H^{i+1}$ is projected onto the fine $H^{i}$ (interpolation) and refined using a local search heuristic such as Kernighan-Lin (KL)~\cite{kernighan1970efficient} or Fiduccia-Mattheyses (FM)~\cite{fiduccia1988linear} (refinement). This retains the global information of the partitioning of the coarse hypergraph while optimizing it locally. Typically, solving a local search subproblem only improves the global solution at the same level.

\section{Related work}

Because HGP is NP-complete~\cite{gary1979computers},  many heuristics and approximations have been developed. The most common practical approach to HGP is the multilevel framework. This section  begins with a brief description of non-multilevel techniques, followed by multilevel ones.

\subsection{Spectral methods}

An important family of non-multilevel techniques is spectral methods. It is necessary to point out that while they can be used as standalone methods, they are also often used within the multilevel framework. As we discussed in  \cref{sec:prereq_multilevel}, the multilevel method for combinatorial optimization problems is a heuristic that incorporates other heuristics as well, such as different similarity concepts and iterative refinement techniques. Spectral hypergraph partitioning generalizes spectral graph partitioning methods to hypergraphs. It usually utilizes the spectral properties of the adjacency matrix. Two main approaches are identified.

The first one is to construct a graph from the hypergraph~\cite{hagen1992new} and then apply spectral graph partitioning methods that are more well-developed~\cite{ng2001spectral,fiedler1973algebraic}. Two of the most common approaches are star and clique expansions. In the case of clique expansion, a hyperedge is replaced by a set of edges that form a complete subgraph for the vertices in the hyperedge. In the case of star expansion, a hyperedge is replaced by a new vertex, which is connected by new edges to all vertices previously contained in the hyperedge.

This approach suffers from an obvious loss of information: when a hyperedge is expanded (i.e., replaced by a clique or a star), its vertices are connected by a number of edges. The information that they are equal members of a hyperedge is lost.
Ihler et al. show~\cite{ihler1993modeling} that even for bi-partitioning there exists no min-cut graph model of a hypergraph. That is, one cannot create a graph whose edge cut is equal to the hyperedge cut in the original hypergraph, if negative weights are not allowed~\cite{ihler1993modeling}. Finally, the hypergraph-to-clique conversion greatly increases the size of the problem. Nevertheless, we point out that despite these limitations, good practical results can still be obtained by using graph models of a hypergraph. 

The second approach is to build hypergraph Laplacian and to study its properties, bypassing the graph representation. This can be done in various ways. Bolla defines an unweighted hypergraph Laplacian matrix and shows the link between its spectral properties and the hypergraph cut~\cite{bolla1993spectra}. Zhou et al. define a Laplacian matrix and show a way to use it for $k$-way partitioning~\cite{zhou2006learning}. Hu et al. argue that Laplacian tensors naturally extend the graph Laplacian matrices to hypergraphs. They describe a Laplacian tensor for an even uniform hypergraph and define algebraic connectivity for it~\cite{hu2012algebraic}. Chan et al. define a Laplacian operator induced by a stochastic diffusion process on the hypergraph  and generalize Cheeger's inequality for it~\cite{chan2016spectral}. However, these recent advances of the spectral approaches, while promising, are not yet well developed for large-scale instances.

\subsection{Multilevel methods}

Most state-of-the-art hypergraph partitioners (such as hMetis2 \cite{karypis1999multilevel}, PaToH \cite{ccatalyurek2011patoh}, Zoltan \cite{devine2006parallel} and Mondriaan \cite{vastenhouw2005two}, to name a few) use a multilevel approach inspired by simplified multigrid and principles of multiscale computing. 

In the coarsening stage of the V-cycle, most hypergraph partitioners use, with some variations, a heuristic that greedily aggregates neighboring vertices, with some preference based on a similarity metric. These similarity metrics are local and usually very simple. The metric used by Mondriaan~\cite{vastenhouw2005two}, hMetis2~\cite{karypis1999multilevel} and Zoltan~\cite{devine2006parallel} 
is \emph{inner product matching}. The inner product of two vertices is defined as the Euclidean inner product of their hyperedge incidence vectors~\cite{devine2006parallel}. In other words, the inner product of two vertices is the number of hyperedges they have in common or, in the weighted case, total weight of those hyperedges. PaToH uses a different metric as a default option, called \textit{absorption matching}:
\[
\text{am}(u,v) = \sum_{e=(u,v)\in E}\frac{1}{|e|-1},
\]
where $|e|$ is the cardinality of an edge $e$. However, there are scenarios where these simple solutions, while computationally efficient and easy to implement, are not very effective and can be improved~\cite{chen2011algebraic}. This is the major reason for us to revise the coarsening strategy for hypergraphs. In the refinement step, all of the aforementioned partitioners use some variation of Fiduccia-Mattheyses~\cite{fiduccia1988linear} or Kernighan-Lin~\cite{kernighan1970efficient} including their advanced efficient implementation introduced in \cite{akhremtsev2017engineering}.

However, there is relatively little research on how to improve the coarsening of a hypergraph. Existing research was motivated mainly by the intuition that a decision made earlier in the coarsening stage has substantial influence on the quality of the final cut: any error or wrong decision would be propagated all the way down the V-cycle and accumulate. Some methods that improve the coarsening include a use of rough set theory~\cite{lotfifar2015multi}, as well as some variations on the greedy heavy matching scheme. A very promising but unfinished attempt to generalize HGP coarsening using algebraic multigrid was published in Sandia Labs Summer Reports \cite{buluc-boman}. Another extension of a multilevel method, namely, the n-level recursive bisection, was introduced in \cite{schlag2016k}. For a more extensive review of hypergraph partitioning algorithms the reader is referred to~\cite{trifunovic2006parallel}.

\section{Algebraic Distance on Hypergraphs}

In this section, we introduce a new distance measure that plays a crucial role in improving the quality of coarsening. To demonstrate its effectiveness, we use Zoltan~\cite{devine2006parallel} as the baseline solver. 

The outline of our approach is as follows. At each level of coarsening, we compute new weights for the hyperedges. These weights are passed to Zoltan's coarsening subroutine, allowing it to use this additional information for making matching decisions. After the matching is computed, the weights are set back to the original hyperedge weights and the multilevel algorithm continues. 
In other words, we leverage the hyperedge weights to pass information on the structure of the hypergraph, derived by our algorithm, to the HGP solver's coarsening scheme. We refer to these weights as \emph{algebraic weights}.
The outline is shown in \cref{fig:flow_chart}.
%This section is structured as follows: it begins by describing the algorithm in detail (\cref{sec:alg_dist_algorithm}), then analyzes its convergence properties (\cref{sec:alg_dist_analysis}) and finally lays out the intuition behind this approach (\cref{sec:alg_dist_intuition}).

\begin{figure}
  \centering
  \includegraphics[width=\textwidth]{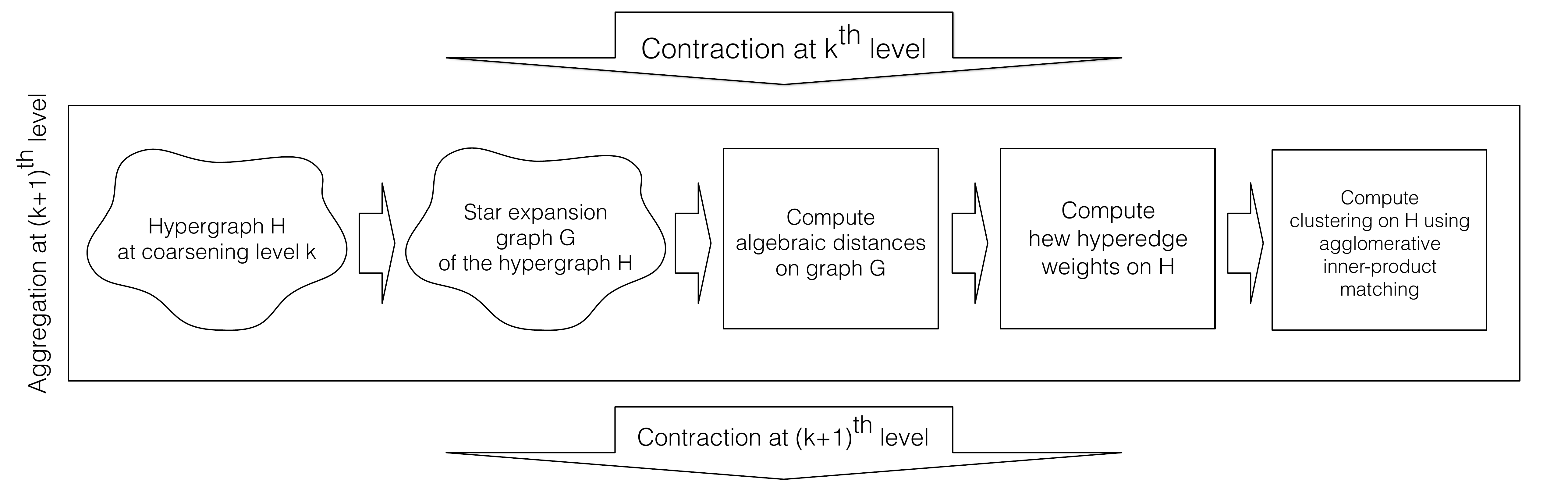}
  \caption{Outline of the algorithm used to aggregate vertices at $k$th coarsening level.}
  \label{fig:flow_chart}
\end{figure}

Discussion of this algorithm would not be complete without a brief description of the coarsening scheme used by Zoltan. Zoltan uses an agglomerative matching technique known as inner-product matching~\cite{catalyurek2009repartitioning}, or called heavy-connectivity clustering in PaToH~\cite{catalyurek1999hypergraph}. In this technique, the vertices are visited in random order. If the visited vertex $v$ is unmatched, an adjacent vertex $u$ with the highest \emph{connectivity} is selected and the current vertex is added to its cluster $C_u$. The \emph{connectivity} is defined as $N_{v,C_u}/W_{v,C_u}$, where $N_{v,C_u}$ is the total weight of the edges connecting $v$ with the vertices in the cluster $C_u$, and $W_{v,C_u}$ is the total weight of the vertices in the candidate cluster~\cite{catalyurek1999hypergraph}.

\subsection{Algorithm}
\label{sec:alg_dist_algorithm}

Recall that the hypergraph is denoted as $H = (V,E)$. We call the star expansion graph $G$. In general, we use primed variables when referring to the elements of $G$ and non-primed for the elements of $H$. At each coarsening level the following happens:

\begin{enumerate}[leftmargin=*]

\item Build the star expansion graph $G = (V',E')$ of  $H$~\cite{sun2008hypergraph}. Each hyperedge is replaced by a new vertex, which is connected to every vertex contained in the original hyperedge, i.e., $V' = V\cup E$ and  $E' = \{(v,h) \mid v\in h, h \in E\}$~\cite{sun2008hypergraph}. This way, each hyperedge is replaced by a ``star'' in the bipartite graph $G$, with all vertices $v\in V$ in one part and  hyperedges $h\in E$ in another (hence the new set of vertices is $V' = V\cup E$).
The weights of the vertices stay the same. The vertices introduced to replace hyperedges are assigned the weights of the corresponding hyperedges they represent, divided by the number of vertices in that hyperedge; i.e.,  
\[ 
w'(v') = w(v), \quad \forall v\in V\subset V', \text{ and}
\]
\begin{equation} \label{eq:2}
w'(h') = \frac{w(h)}{|h|}, \quad \forall h\in E \subset V'.
\end{equation}

The edges $e'\in E'$ of the star expansion $G$ are left unweighted (we can assume they all have weight $1$). We refer to the vector of weights in the graph $G$ as $\textbf{w}'\in\mathbb R^{|V'|} $ and that in the hypergraph $H$ as $\textbf{w}\in\mathbb R^{|V|+|E|}$, to distinguish the two cases. Note that since $V' = V\cup E$, these vectors are of the same size. 

\item Compute algebraic coordinates of the vertices in star expansion graph $G$ by using Jacobi over-relaxation (JOR, a stationary iterative relaxation), beginning with random initialization. The iterative process is repeated for several random initial vectors.

For each random vector, all coordinates are initialized with random values from the uniform distribution on $(-0.5, 0.5)$. We denote the algebraic coordinates as a 2-dimensional array $X$, where $x[r][v']$ is the algebraic coordinate for the vertex $v'\in V' = V\cup E$ and the $r$th random initial vector. Then, for a certain number of iterations (denoted in pseudocode as $num\_iter$), JOR is performed on all $v'\in G$ as follows. The vertices are visited in random order. After each iteration, the algebraic coordinates are rescaled such that the smallest algebraic coordinate is equal to $-0.5$ and the largest to $0.5$. The iteration scheme can be written as follows
\begin{equation} \label{eq:1}
\begin{split}
x^*[r][v'] = \frac{\mathlarger{\mathlarger{\sum}}\limits_{\forall u'\in V':(u',v')\in E'}w'(u')\cdot  x^{i-1}[r][u']}{\mathlarger{\mathlarger{\sum}}\limits_{\forall u'\in V':(u',v')\in E'}w'(u')} \\
x^{**}[r][v']=\omega \cdot  x^*[r][v'] + (1-\omega)x^{i-1}[r][u'],
\end{split}
\end{equation}
where $\omega$ is a relaxation factor. It is used in the same way as in Successive overrelaxation~\cite{briggs2000multigrid}, to make the convergence more stable. Weighting, rescaling and the denominator in \cref{eq:1} are introduced in part to prevent algebraic coordinates from converging to machine precision, i.e., becoming so close that they are no longer distinguishable. In our experiments, nondistinguishable coordinates tend to occur on hypergraphs whose hyperedges have high cardinality, since vertices contained in these hyperedges very quickly ``pull'' each other together.

Denote the smallest (``leftmost'' on the real line) algebraic coordinate before rescaling as $l$ and the largest (``rightmost'') as $r$. Then, the rescaling is done as follows:
\[
x^{i}[r][v']=\frac{x^{**}[r][v']-l}{r-l}-0.5.
\]

\item Define the algebraic weight of a hyperedge $h\in E$ as one over the algebraic distance between two farthest apart vertices in $h$, maximized over all random vectors: 
\[
alg\_weight(h) = \frac{1}{\max\limits_{\forall r} \max\limits_{u,v\in h} |x[r][v]-x[r][u]| }.
\]

Compute the final weights $\tilde{\textbf{w}}$ to be passed to Zoltan from the original hyperedge weights $\textbf{w}$ as follows
\[
\tilde{w}(h) = w(h)\cdot\dfrac{alg\_weight(h)}{\mathlarger{\sum}\limits_{h^*\in E}alg\_weight(h^*) / |E|}.
\]

Hyperedge weights are multiplied by the ratio between the computed algebraic distance for this hyperedge and the average algebraic weight for all hyperedges. Note that here $w(h)$ is the weight of the original hyperedge, before the star expansion.

\item Pass these new weights $\tilde{\textbf{w}}$ to Zoltan's~\cite{devine2006parallel} agglomerative inner product matching. After the matching, all weights are reset back to the ones before the star expansion (i.e., back to $\textbf{w}$) and the multilevel process continues.

\end{enumerate}

The pseudocode for computing the algebraic weights (steps 1--3) is presented in \cref{alg:weights}.

\begin{algorithm2e}
    \SetKwInOut{Input}{Input}
    \SetKwInOut{Output}{Output}

    \Input{Relaxation factor $\omega$, number of random vectors $R$ , number of iterations $num\_iter$}
    \Output{algebraic weights $alg\_weight$}
    \For{$r=1:\text{R}$}{
    	Randomly initialize $x[r]$\;
        \For{$k=1:num\_iter$}{
        	\tcp{Perform iteration sweep over all vertices}
        	\For{$v'\in V'$}{
				$x^*[r][v'] = \frac{\mathlarger{\mathlarger{\sum}}\limits_{\forall u'\in V':(u',v')\in E'}w'(u')\cdot  x^{i-1}[r][u']}{\mathlarger{\mathlarger{\sum}}\limits_{\forall u'\in V':(u',v')\in E'}w'(u')}$\;
				$x^{**}[r][v']=\omega \cdot  x^*[r][v'] + (1-\omega)x^{i-1}[r][u']$\;
            }
            \For{$v\in V'$}{
            	\tcp{Rescale}
                $l=\min\limits_{u\in V'}{x[r][u]}$\;
                $r = \max\limits_{u\in V'}{x[r][u]}$\;
            	$x^{i}[r][v']=\frac{x^{**}[r][v']-l}{r-l}-0.5$\;
            }
        }
    }
    \For{h in E}{
    	\tcp{Compute algebraic weights for hyperedges}
    	$alg\_weight(h) =1\, /\, \max\limits_{\forall r} \max\limits_{u,v\in h} |x[r][v]-x[r][u]| $\;
    }
    \caption{Computing algebraic weights}
    \label{alg:weights}
\end{algorithm2e}

\subsection{Convergence analysis}
\label{sec:alg_dist_analysis}

\cref{alg:weights} is an iterative process
that computes $x[r][v']$ for all vertices $v'\in V'$ and random initial vector numbers $r$. To analyze the convergence of this process and understand the properties of the computed algebraic weights, we need additional notation. Let $x^{(i)}\in\real^{|V'|}$ denote the $i$th iterate of the vector $x[r]$ for the $r$th random initial vector. Since the same iterative process is performed for all vectors of random initial coordinates, we will perform the analysis on only one (arbitrary) vector and hence omit the dependence on $r$.

Let $A\in\real^{|E|\times|V|}$ be the incidence matrix of the hypergraph; that is, $A_{ij}=1$ if vertex $j$ belongs to the hyperedge $i$ and $0$ otherwise. Let $S^v\in\real^{|V|\times|V|}$ and $S^h\in\real^{|E|\times|E|}$ be diagonal matrices such that
\[
S_{jj}^v=w(v_j) \quad\text{and}\quad S_{ii}^h=\frac{w(h_i)}{|h_i|},
\]
where $|h_i|$ denotes the cardinality of the $i$th hyperedge (same as in \cref{eq:2}). Define the following matrix
\[
W=
\begin{bmatrix}
0 & A^TS^{h} \\
AS^{v} & 0
\end{bmatrix} 
\]
and let $D$ be the diagonal matrix with elements $D_{jj}=\sum_iW_{ij}$. For convenience of analysis, we decompose $D$ as
\[
D=
\begin{bmatrix}
D^{v} & 0 \\
0 & D^{h}
\end{bmatrix}
\]
with $D^{v}\in\real^{|V|\times|V|}$ and $D^{h}\in\real^{|E|\times|E|}$, such that the block sizes of $D$ are compatible with those of $W$. Note that $W$ is asymmetric.

It is not hard to see that \cref{alg:weights} computes the following update:
\[
x^{(i)}=\frac{1}{r-l}\Big[\underbrace{\omega D^{-1}Wx^{(i-1)}+(1-\omega)x^{(i-1)}}_{{x^{*}}^{(i-1)}}\Big]
-\frac{r+l}{2(r-l)}\bm{1}
\]
where $\bm{1}$ is the vector of all ones, and $r$ and $l$ are the maximum and the minimum of the elements in ${x^{*}}^{(i-1)}$, respectively (see \cref{alg:weights}). Then, we simplify the update formula as
\begin{equation}\label{eq:xk}
x^{(i)}=\alpha^{(i-1)}Hx^{(i-1)}+\beta^{(i-1)}\bm{1},
\end{equation}
where
\[
H=\omega D^{-1}W+(1-\omega)I,\quad
\alpha^{(i-1)}=\frac{1}{r-l},\quad\text{and}\quad
\beta^{(i-1)}=-\frac{r+l}{2(r-l)}.
\]

The convergence of the iteration depends on the properties of the iteration matrix $H$. We first note that the matrix $D^{-1}W$ is diagonalizable with real eigenvalues.

\begin{theorem}\label{thm:dw_diagonalizability}
Let $(D^{h})^{-1/2}(S^{h})^{1/2}A(D^{v})^{-1/2}(S^{v})^{1/2}$ admit the following singular value decomposition:
\[
(D^{h})^{-1/2}(S^{h})^{1/2}A(D^{v})^{-1/2}(S^{v})^{1/2} = \sum_{i=1}^r \sigma_iu_iz_i^T.
\]
Then, $D^{-1}W$ has a zero eigenvalue of multiplicity $|V|+|E|-2r=|V'|-2r$. Moreover, for the nonzero eigenvalues, 
\begin{equation}\label{eqn:DW.eig1}
D^{-1}W\begin{bmatrix} (S^{v})^{-1/2}(D^{v})^{-1/2}z_i \\ (S^{h})^{-1/2}(D^{h})^{-1/2}u_i \end{bmatrix}=
\sigma_i\begin{bmatrix} (S^{v})^{-1/2}(D^{v})^{-1/2}z_i \\ (S^{h})^{-1/2}(D^{h})^{-1/2}u_i \end{bmatrix}
\end{equation}
and
\begin{equation}\label{eqn:DW.eig2}
D^{-1}W\begin{bmatrix} -(S^{v})^{-1/2}(D^{v})^{-1/2}z_i \\ (S^{h})^{-1/2}(D^{h})^{-1/2}u_i \end{bmatrix}=
-\sigma_i\begin{bmatrix} -(S^{v})^{-1/2}(D^{v})^{-1/2}z_i \\ (S^{h})^{-1/2}(D^{h})^{-1/2}u_i \end{bmatrix},
\quad i=1,\ldots r.
\end{equation}
\end{theorem}

\begin{proof}
The identities \cref{eqn:DW.eig1} and \cref{eqn:DW.eig2} are straightforward to verify based on the singular value decomposition. On the other hand, if $\lambda$ is an eigenvalue of $D^{-1}W$ with corresponding eigenvector $\left[\begin{smallmatrix} c \\ d \end{smallmatrix}\right]$, then,
\[
(D^{v})^{-1}A^TS^{h}d=\lambda c \quad\text{and}\quad (D^{h})^{-1}AS^{v}c=\lambda d,
\]
which implies that
\[
[(D^{h})^{-1/2}(S^{h})^{1/2}A(D^{v})^{-1/2}(S^{v})^{1/2}]^T[(D^{h})^{1/2}(S^{h})^{1/2}d]=\lambda[(D^{v})^{1/2}(S^{v})^{1/2}c],
\]
and
\[
[(D^{h})^{-1/2}(S^{h})^{1/2}A(D^{v})^{-1/2}(S^{v})^{1/2}][(D^{v})^{1/2}(S^{v})^{1/2}c]=\lambda[(D^{h})^{1/2}(S^{h})^{1/2}d].
\]
In form, these two equalities define all nonzero singular values $|\lambda|$\linebreak of $(D^{h})^{-1/2}(S^{h})^{1/2}A(D^{v})^{-1/2}(S^{v})^{1/2}$. Thus, if $\lambda$ is not equal to any of the $\pm\sigma_i$'s, then $\lambda$ must be zero.
\end{proof}

Second, note that the diagonal of the matrix $D^{-1}W$ is zero and each row is nonnegative, summing to one. Hence, by Gershgorin's circle theorem, the eigenvalues of $D^{-1}W$ must lie between $-1$ and $1$. Indeed, the two ends of this interval are attainable.

\begin{theorem}
The spectral radius of $D^{-1}W$ is $1$. In particular, there exists a pair of eigenvalues $\pm1$, and $\bm{1}$ is an eigenvector associated with an eigenvalue $+1$.
\end{theorem}

\begin{proof}
By the definition of $W$ and $D$, we have $D^{-1}W\bm{1}=\bm{1}$. Therefore, $D^{-1}W$ has an eigenvalue $+1$ with eigenvector $\bm{1}$. From \cref{thm:dw_diagonalizability} we know that $D^{-1}W$ must also have an eigenvalue $-1$. Then, because all the eigenvalues of $D^{-1}W$ lie between $-1$ and $+1$ by the Gershgorin's circle theorem, its spectral radius is $1$.
\end{proof}

Third, recall that the star expansion graph is a bipartite graph, with one part containing the vertices of the original hypergraph and the other part the hyperedges. Rather than defining vertex weights as in \cref{eq:2} and leaving the edges unweighted, alternatively, one may view $W$ as the weighted adjacency matrix of a graph with exactly the same vertex-edge structure, but the edges are weighted according to $W$ and the vertices are left unweighted. Then, the block and the sparsity structure of $W$ indicates that the directed edges of this bipartite graph (in the alternative view) always come in pairs with opposite directions. Hence, we may ignore the directions and treat the graph undirected, if at all convenient for analyzing graph connectivity. Then, the directed version of the graph is strongly connected if and only if the undirected version is connected, although the edge weights are asymmetric in the directed version. Consequently, the number of strongly connected components of the directed version is equal to the number of connected components of the undirected version. Such a view allows the application of the Perron--Frobenius theorem for exploring the extreme eigenvalues of $D^{-1}W$.

\begin{theorem}
The eigenvalues of $D^{-1}W$ equal to $+1$ are all simple. The number of such eigenvalues is equal to the number of strongly connected components of the directed version of the bipartite graph (or equivalently, the number of connected components of the undirected version of the graph).
\end{theorem}

\begin{proof}
If the directed graph is strongly connected, then $D^{-1}W$ is irreducible. Hence, by the Perron--Frobenius theorem, there exists an eigenvalue equal to the spectral radius $1$. This eigenvalue is simple and is unique. If the graph has $C$ strongly connected components, then $D^{-1}W$ may be symmetrically permuted into a $C\times C$ block-diagonal matrix, where each block corresponds to one strongly connected component. In such a case, $D^{-1}W$ has $C$ eigenvalues equal to $1$, all of which are simple.
\end{proof}

The above three theorems reveal the beautiful symmetry of the eigenvalues of $D^{-1}W$. They are real, come in pairs, and straddle around zero (except those being exactly zero). The number of nonzero pairs is equal to the rank of the incidence matrix $A$. The eigenvalues all lie inside $[-1,1]$. Moreover, there exists (at least) one pair attaining exactly $\pm1$, and the number of such pairs is equal to the number of connected components of the star expansion graph when viewed as undirected.

With this knowledge, we see that there is a one-to-one correspondence between the eigenpairs of the iteration matrix $H$ and those of $D^{-1}W$. Specifically, denote by $(\mu_i,\phi_i)$ an eigenpair of $H$. Then,
\[
H\phi_i=\mu_i\phi_i 
\quad\Leftrightarrow\quad
D^{-1}W\phi_i=\frac{\mu_i-1+\omega}{\omega}\phi_i,
\quad 0<\omega<1.
\]
Because the eigenvalues of $D^{-1}W$ are symmetric around $0$ with range $[-1,1]$, the eigenvalues $\mu_i$ of $H$ are symmetric around $1-\omega$ with range $[1-2\omega,1]$. Because $\omega$ is strictly less than $1$, the largest eigenvalue of $H$ is $1$, which is equal to the largest eigenvalue in magnitude. The multiplicity of this eigenvalue is equal to the number of connected components of the undirected version of the star expansion graph. For simplicity of analysis, we will assume from now on that the graph is connected. Then, we have an ordering of the eigenvalues according to their magnitude:
\[
1=\mu_1>|\mu_2|\ge|\mu_3|\ge\cdots\ge|\mu_{|V'|}|,
\]
with $\phi_1=\bm{1}$ being an (unnormalized) eigenvector associated with $\mu_1$. Note the use of the strictly-greater-than sign $>$ and the greater-than-or-equal-to sign $\ge$. In particular, the second eigenvalue, in magnitude, must be strictly less than $1$. In what follows, we will present a result that relates the difference between the elements of $x^{(k)}$ to that between the corresponding elements of some vector in the eigensubspace spanned by one or a few eigenvectors, including $\phi_2$. This eigensubspace depends on how many eigenvalues are equal to $|\mu_2|$ in magnitude. If only one, that is, $|\mu_2|>|\mu_3|$, then the eigensubspace is spanned by only $\phi_2$. However, if more than one, then let us assume that $|\mu_2|=|\mu_3|=\cdots=|\mu_t|>|\mu_{t+1}|$. Such a case includes two subcases:
\begin{description}
\item[case 1:] $\mu_2=\mu_3=\cdots=\mu_t$; and
\item[case 2:] $\mu_2$, $\mu_3$, \ldots, $\mu_t$ are not all equal.
\end{description}
In each subcase, the vector is some linear combination of $\phi_2$, \ldots, $\phi_t$.

\begin{theorem}\label{thm:convergence2}
Assume that the undirected version of the star expansion graph is connected. Let the initial iterate $x^{(0)}$ be expanded in the eigen-basis of $H$ as
\[
x^{(0)}=a_1\phi_1+a_2\phi_2+\cdots+a_{|V'|}\phi_{|V'|}.
\]
\begin{enumerate}[label=(\roman*),leftmargin=*]
\item If $\mu_2=\mu_3=\cdots =\mu_t$ and $|\mu_t|>|\mu_{t+1}|$ for some $t\geq 2$, and if $a_2$, \ldots, $a_t$ are not all zero, then for any pair $i,j$,
\[
\lim_{k\to\infty}\frac{(x^{(k)})_i-(x^{(k)})_j}{\alpha^{(0)}\alpha^{(1)}\cdots\alpha^{(k-1)}\mu_2^k}=
\xi_i-\xi_j,
\]
where
\[
\xi=a_2\phi_2+\cdots+a_t\phi_t.
\]
\item If $|\mu_2|=|\mu_3|=\cdots =|\mu_t|>|\mu_{t+1}|$ for some $t\geq 3$ where $\mu_2$, $\mu_3$, \ldots, $\mu_{t}$ are not all equal, and if $a_2$, \ldots, $a_t$ are not all zero, then for any pair $i,j$,
\[
\lim_{k\to\infty}\frac{(x^{(2k+p)})_i-(x^{(2k+p)})_j}{\alpha^{(0)}\alpha^{(1)}\cdots\alpha^{(2k+p-1)}\mu_2^{2k+p}}=(\eta_p)_i-(\eta_p)_j,
\]
where
\[
\eta_p = a_2\phi_2 + a_3(\mu_3/\mu_2)^p\phi_3 + \cdots + a_t(\mu_t/\mu_2)^p\phi_t,
\quad p=0,1.
\]
\end{enumerate}
\end{theorem}

\begin{proof}
For notational convenience, let $n=|V'|$. Write the diagonalization of $H$ as $H=\Phi\Lambda\Phi^{-1}$, where $\Phi=[\phi_1,\ldots,\phi_n]$ and $\Lambda=\diag(\mu_1,\ldots,\mu_n)$.  Also write $x^{(0)}=\Phi a$, where $a=[a_1,\ldots,a_n]^T$. Then, expanding~\eqref{eq:xk} $k$ times, we have
\begin{equation}
\begin{split}
x^{(k)}&=\alpha^{(0)}\ldots\alpha^{(k-1)}H^kx^{(0)}
+\\
&+(\beta^{(0)}\alpha^{(1)}\ldots\alpha^{(k-1)}H^{k-1}
+\beta^{(1)}\alpha^{(2)}\ldots\alpha^{(k-1)}H^{k-2}
+\cdots
+\beta^{(k-1)}H^0)\bm{1}.
\end{split}
\end{equation}

Because $\bm{1}$ is an eigenvector of $H$ corresponding to eigenvalue $1$, the term in the second line is equal to $\gamma^{(k)}\bm{1}$ with $\gamma^{(k)}=\beta^{(0)}\alpha^{(1)}\ldots\alpha^{(k-1)}
+\beta^{(1)}\alpha^{(2)}\ldots\alpha^{(k-1)}
+\cdots
+\beta^{(k-1)}$.
For the first term on the right of the first line,
\[
H^kx^{(0)}=\Phi\Lambda^ka
=a_1\phi_1+a_2\mu_2^k\phi_2+\cdots+a_n\mu_n^k\phi_n,
\]
with $\phi_1=\bm{1}$. Therefore,
\[
x^{(k)}
=\alpha^{(0)}\ldots\alpha^{(k-1)}(a_1\bm{1}+a_2\mu_2^k\phi_2+\cdots+a_n\mu_n^k\phi_n)
+\gamma^{(k)}\bm{1},
\]
and hence
\[
\frac{(x^{(k)})_i-(x^{(k)})_j}{\alpha^{(0)}\ldots\alpha^{(k-1)}\mu_2^k}=(e_i-e_j)^T[a_2\phi_2+a_3(\mu_3/\mu_2)^k\phi_3+\cdots+a_n(\mu_n/\mu_2)^k\phi_n].
\]
\begin{enumerate}[label=(\roman*),leftmargin=*]
\item When $k$ is large, the term inside the square bracket is dominated by 
\[
a_2\phi_2+\cdots+a_t\phi_t,
\]
because all other $\mu_i$'s are smaller than $\mu_2$ in magnitude. Thus, when $k\to\infty$, only this term remains, hence the result.

\item Similar to the above case,
\begin{multline*}
\frac{(x^{(2k+p)})_i-(x^{(2k+p)})_j}{\alpha^{(0)}\ldots\alpha^{(2k+p-1)}\mu_2^{2k+p}}
=(e_i-e_j)^T[a_2\phi_2+a_3(\mu_3/\mu_2)^{2k+p}\phi_3+\cdots\\
\cdots+a_n(\mu_n/\mu_2)^{2k+p}\phi_n].
\end{multline*}

When $k\to\infty$, in the square bracket only the term
\[
a_2\phi_2+a_3(\mu_3/\mu_2)^{2k+p}\phi_3+\cdots+a_t(\mu_t/\mu_2)^{2k+p}\phi_t
\]
remains and the rest vanishes. Because the eigenvalues $\mu_2$, \ldots, $\mu_t$ are all real, the ratios $\mu_3/\mu_2$, \ldots $\mu_t/\mu_2$ can only be $\pm1$. Hence, taking square, the ratios all become $1$. We thus obtain the result in the theorem.
\end{enumerate}
\end{proof}

Informally speaking, the above theorem states that in the limit, the difference between two elements of the iterate vector $x^{(k)}$ is proportional to that between the corresponding elements of some vector $\xi$. When $|\mu_2|$ is strictly greater than $|\mu_3|$, this vector $\xi$ may be taken to be the eigenvector $\phi_2$. When there exist more than one eigenvalue equal to $\mu_2$, say, $\mu_2=\cdots=\mu_t$ for some $t>2$, then $\xi$ is a linear combination of $\phi_2$, \ldots, $\phi_t$, where the coefficients of the combination are the expansion coefficients of the initial iterate $x^{(0)}$ along the eigen-basis of the iteration matrix $H$. However, when there exist more than one eigenvalue whose magnitude is equal to $|\mu_2|$ but these eigenvalues are not all the same, then the situation is slightly complicated. Every other iterate $x^{(k)}$ form a subsequence and the limiting behaviors of these two interleaving subsequences correspond to two different vectors, which are $\eta_0$ and $\eta_1$ defined in the theorem. Both of them are a linear combination of the eigenvectors $\phi_2, \ldots, \phi_t$. For~$\eta_0$, the coefficients of the combination are the expansion coefficients of $x^{(0)}$ along the eigen-basis of $H$; for $\eta_1$, some of these coefficients flip signs.

The existence of these various cases is owing to the different choices of $\omega$. Recall that the eigenvalues of $D^{-1}W$ are symmetrically distributed around zero. Then, the eigenvalues of $H=\omega D^{-1}W+(1-\omega)I$ are symmetrically distributed around $1-\omega$, with the smallest one being $1-2\omega$. We will use $\sigma_2>0$ to denote the second largest eigenvalue of $D^{-1}W$, which is consistent with the notation in Theorem~\ref{thm:dw_diagonalizability}. By the connected-graph assumption, $\sigma_2<\sigma_1=1$. We also have that the multiplicity of $1-2\omega$ is one, by the same assumption. Furthermore, the second largest eigenvalue of $H$ is $\omega\sigma_2+1-\omega$.
\begin{enumerate}[label=(\alph*),leftmargin=*]
\item If $\omega>\frac{2}{3-\sigma_2}$, then $1-2\omega$ is negative and it has a larger magnitude than does $\omega\sigma_2+1-\omega$ . In such a case, $\mu_2=1-2\omega$ and $|\mu_2|$ is strictly greater than $|\mu_3|$. This scenario corresponds to the case (i) of Theorem~\ref{thm:convergence2} with $t=2$.

\item If $\omega\le\frac{1}{2}$, then $1-2\omega$ is nonnegative, and hence all eigenvalues of $H$ are nonnegative. In such a case, $\mu_2=\omega\sigma_2+1-\omega$. This scenario also corresponds to the case (i) of Theorem~\ref{thm:convergence2}, because if there are more than one eigenvalue whose magnitude is equal to $\mu_2$, then the nonnegativity implies that these eigenvalues must be the same as $\mu_2$. The multiplicity of $\mu_2$ is equal to the multiplicity of $\sigma_2$. It could happen that either the multiplicity of $\sigma_2$ is one, in which case $t=2$; or the multiplicity of $\sigma_2$ is greater than one, in which case $t>2$.

\item If $\omega=\frac{2}{3-\sigma_2}$, then $1-2\omega$ is negative, $\omega\sigma_2+1-\omega$ is positive, and the two have the same magnitude. This scenario corresponds to the case (ii) of \cref{thm:convergence2}. Whether $t=3$ or $t>3$ depends on the multiplicity of $\omega\sigma_2+1-\omega$ (equivalently that of $\sigma_2$), because the multiplicity of $1-2\omega$ is one. If the multiplicity of $\sigma_2$ is one, then only two eigenvalues have magnitude equal to $|\mu_2|$, and hence $t=3$; otherwise, $t>3$.

\item If $\frac{1}{2}<\omega<\frac{2}{3-\sigma_2}$, then $1-2\omega$ is negative and $\omega\sigma_2+1-\omega$ is greater than the magnitude of $1-2\omega$. This scenario is similar to that of the above item (b), except that not all the eigenvalues of $H$ are nonnegative. All other properties are otherwise the same.
\end{enumerate}

\subsection{Mutually influenced model}
\label{sec:alg_dist_intuition}

In the standard graph case, the edge weights provide a first-order measurement of vertex\linebreak similarity/distance. This measure is applicable to only adjacent vertices. Hence, for a pair of vertices that are not adjacent, global information of the graph must be incorporated for extending the measurement. The \emph{algebraic distance on graphs}~\cite{chen2011algebraic} defines algebraic coordinates for every vertex through an iterative process similar to that in the present work. In the limit, the coordinate difference, which serves the notion of ``distance,'' is proportional to the difference of the corresponding elements of some vector $y$. Specifically, if we let $w_{ij}$ be the weights of a pair of vertices $ij$ in the graph, then the elements of the vector $y$ satisfy an equilibrium state
\begin{equation}\label{eqn:model.graph}
y_i=\gamma y_i+\sum_{j}\frac{w_{ij}}{d_i}y_j,
\end{equation}
where $d_i=\sum_{j}w_{ij}$ is used for normalization and $0<\gamma<1$ is a constant factor for all $i$. The algebraic coordinates do not need to coincide with the $y_i$'s; it suffices for their differences to be proportional. When one treats the vertices to be entities of a mutually influenced environment, then the value $y_i$ of each entity is composed of two components according to~\eqref{eqn:model.graph}: a portion of itself ($\gamma y_i$) and a normalized weighted contributions of its neighbors ($\sum_{j}\frac{w_{ij}}{d_i}y_j$). If $y_i$ denotes the amount of information stored at vertex $i$, \cref{eqn:model.graph} essentially signifies a global equilibrium of the flow of information. Such a view is used to interpret the global similarity of vertices through neighborhoods. In short, two vertices are similar if their neighborhoods are similar, because the common factor $\gamma$ is constant and similarity relies on the neighboring weights $w_{ij}$.

The present work extends this notion to hypergraphs. In such a setting, the hyperedges cover not just a pair $ij$ of vertices, but rather, a subset of any size (excluding of course empty sets and singletons). Hence, a proxy of the pairwise environment is the star expansion graph $G$, whose vertex set $V'$ includes not only the original vertices $V$ of the hypergraph, but also the hyperedges $E$. Naturally, the edges of $G$ come from the containment relation between $V$ and $E$; that is, $v\in V$ and $h\in E$ are connected if and only if $v\in h$. Therefore, the edge weights of $G$ come from the weights in the original hypergraph, properly scaled. It is important to note that the scaling is not symmetric, because potentially one hyperedge may contain a large number of vertices. Thus, whereas the weights for the directed edge from $v$ to $h$ are $w(v)$ without scaling, those for the edges from $h$ to $v$ are $w(h)$ scaled by $|h|$.

Then, the \emph{algebraic distance on hypergraphs} enjoys the same equilibrium\linebreak model~\cref{eqn:model.graph}. From \cref{thm:convergence2} and the subsequent discussions, we know that $\gamma=(1-\mu_2)/\omega$, where $\mu_2$ is the second largest eigenvalue in magnitude of the iteration matrix $H$, and $y$ is a vector in the subspace spanned by the eigenvectors associated with eigenvalues whose magnitude are equal to $|\mu_2|$. In practice, $\omega$ is often set to be $1/2$, then $\mu_2$ is positive\footnote{We assume the nondegenerate case where $\mu_2$ is nonzero. Otherwise, all eigenvalues except for the largest one is zero.} and all other eigenvalues having the same magnitude must be equal to $\mu_2$. If furthermore the multiplicity of $\mu_2$ is $1$, then the subspace is spanned by the eigenvector $\phi_2$ only and hence $y$ can be taken to be $\phi_2$. Note that the multiplicity of eigenvalue $\mu_2$ is the same as the multiplicity of the singular value $\sigma_2$ of the matrix $(D^{h})^{-1/2}(S^{h})^{1/2}A(D^{v})^{-1/2}(S^{v})^{1/2}$ in \cref{thm:dw_diagonalizability}.

Finally, when extending the notion of the \emph{algebraic distance} to hypergraphs, we have to take into account the non-pairwise nature of the relations between the vertices. While in pairwise graph setting it is natural to use the simple difference between two vertices' algebraic coordinates as the measure of similarity between them, in non-pairwise hypergraph setting this approach has to be extended. In the hypergraph case we instead assume that the hyperedge is only as important as two most dissimilar vertices in it. In hypergraph partitioning terms, this means that we want to avoid cutting the hyperedges that contain only similar vertices. Therefore we define the algebraic weight of the hyperedge as:
\[
alg\_weight(h) = \frac{1}{\max\{|x[v]-x[u]| \mid u,v\in h\}}
\]
This way, we penalize cutting the hyperedges that contain only similar vertices, with the notion of similarity defined according to the mutually influenced model.

\section{Experimental Results}

In this section, we first illustrate the empirical convergence of~\cref{alg:weights} and then compare the quality of the partitioning produced by our algorithm with those by other state-of-the-art hypergraph partitioners. Our major goal is to study the effectiveness of the algebraic distance on hypergraph partitioning by introducing it into a coarsening scheme. We note that these results can clearly be further improved by using more advanced refinement techniques which are beyond the scope of this work.

\subsection{Convergence}

The speed of the convergence of algebraic weights depends on the gap between the second largest eigenvalue in magnitude, $|\mu_2|$, of the iteration matrix $H$ and the next eigenvalue with a different magnitude, denoted as $|\mu_{t+1}|$ in the preceding section. Estimating this gap is no less expensive than computing the corresponding eigenvectors. Hence, in practice, we use the squared sine of the angle between two iterates $x^{(k)}$ and $x^{(k+1)}$:
\[
 1- \left\langle \frac{x^{(k)}}{||x^{(k)}||},\frac{x^{(k+1)}}{||x^{(k+1)}||} \right\rangle^2,
\]
to measure how parallel the two iterates are, as a proxy of convergence test.

The two consecutive iterates generally become parallel very quickly. In~\cref{fig:convergence} we pick five hypergraphs 
and plot the squared sine for the first 50 iterations. These hypergraphs represent different sizes (from $2426-by-3602$ to $103631-by-395979$) and different origins (from social networks to circuit simulation) of hypergraphs in the benchmark. One sees that the value is indistinguishable from zero after $10$ to $20$ iterations. Such a phenomenon is typical to our experience and we generally set $num\_iter$ comparable to these numbers.

\begin{figure}
  \includegraphics[width=\textwidth]{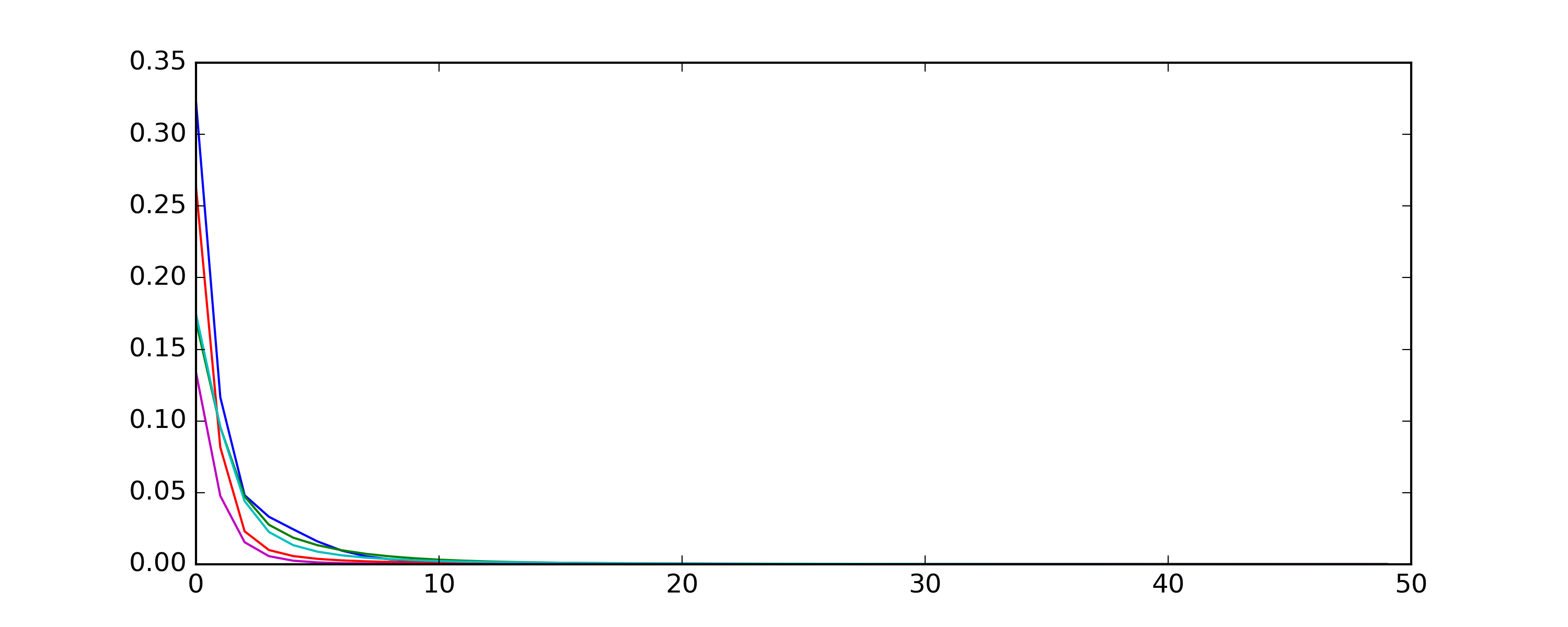}
  \caption{Squared sine of the angle between $x^{(k)}$ and $x^{(k+1)}$ as a function of the iteration number $k$.}\label{fig:convergence}
\end{figure}

It is worth noting that the parallelism of two consecutive iterates does not necessarily mean true convergence. As noted by Chen and Safro~\cite{chen2011algebraic}, the eigenvalue gap may be so small that it might take a huge number of iterations for the difference 
\[
(x^{(k)})_i-(x^{(k)})_j
\]
(after scaling) to get close enough to $\xi_i-\xi_j$. The parallelism before convergence corresponds to a transient state wherein the change of the iterates is small. It turns out that in practice, a transient state is sufficient for the algebraic distance to be useful for coarsening.

\subsection{Cut improvements}

We implement our algorithm with algebraic distances by using the Zoltan~\cite{devine2006parallel} package of the Trilinos Project~\cite{heroux2005overview}. Zoltan is an open-source toolkit of parallel combinatorial scientific computing algorithms~\cite{devine2006parallel} designed to optimize load balancing. It includes a hypergraph partitioning algorithm called PHG (Parallel HyperGraph partitioner). We augment Zoltan's PHG partitioner with algebraic distances as described in the preceding section. Our algorithm is called \algname. All comparisons with Zoltan are indeed comparisons with Zoltan's PHG partitioner. In this study we use Zoltan in the serial mode. In addition to comparing \algname to Zoltan, we compare it with two other state-of-the-art partitioners, namely hMetis2~\cite{karypis1999multilevel} and PaToH v3.2~\cite{ccatalyurek2011patoh}.

PaToH is used as a plug-in for Zoltan, as described in Zoltan's User Guide~\cite{zoltanuserguide}. We run PaToH with two different parameters: default (here denoted as PaToH-D) and quality (PaToH-Q). hMetis2 is used to directly optimize $k$-way partitioning. All parameters are set as default: greedy first-choice scheme for coarsening, random $k$-way refinement, and min-cut objective function.

Each algorithm is run 10 times and the smallest cuts over all runs are compared. In the experiments, standard deviation of the cuts is usually small ($<5\%$). Interestingly, in a small number of hypergraphs, we observe that the distribution of cuts is bimodal (i.e., the partitioner produces cuts close to either one of two modes). In this case, the standard deviation is high; however, within each mode, the deviation is low. Such behavior demonstrates that in certain settings the solvers cannot escape false local attraction basins obtained at the coarse levels. This hints that the current state of hypergraph partitioning solvers is still far from being optimal and there is a lot of space for improvement. 

Because the implementation of our algorithm has not been optimized, the run time of \algname is on average two to four times longer than that of Zoltan (see~\cref{fig:time}). However, since the iteration scheme is easily parallelizable (for example, among different random vectors, the iteration is pleasingly parallel as well as parallelization of Jacobi-based relaxations has been studied and used in many works), the overhead of computing the algebraic distances may be minimized and thus the run time may be made similar to that of Zoltan. Moreover, the algebraic distance is, in fact, several iterations of Jacobi over-relaxation whose parallelization has been studied a lot.

We also have to point out the lack of progress in coarsening techniques in the recent years, with all major hypergraph partitioning packages using the same approach, making our contribution more valuable. While the way we select vertices to be merged together during coarsening is indeed more expensive than traditional approaches, it introduces very low overhead (\textit{constant} number of passes over vertices) and shows great improvement in the area where state-of-the-art has not changed for a long time.

The different algorithms are compared on three groups of hypergraphs: big-bench, SNAP-bench, and social-networks-bench. 

The first two groups (big-bench and SNAP-bench) are generated from matrices obtained from the University of Florida Sparse Matrix Collection~\cite{davis2011university} by using a row-net model: a vertex $i$ belongs to hyperedge $j$ if there is a non-zero element on the intersection of $i$th column and $j$th row. Big-bench contains 443 matrices; many of them are derived from  optimization problems. SNAP-bench consists of 42 matrices from the SNAP (Stanford Network Analysis Platform) Network Data Sets~\cite{snapnets}.

The third group of the benchmark, social-networks-bench (12 hypergraphs), has two parts. The first part contains two networks with known communities (youtube and flickr) from the IMC 2007 Data Set~\cite{mislove-2007-socialnetworks}. The hypergraphs are constructed in the following way: each vertex represents a user and each hyperedge represents a community. After generating the hypergraph, isolated vertices are removed. 
The second part contains ``similar hypergraphs.'' They are generated by using the following pipeline: a graph of pairwise links is obtained for the same dataset and then a similar graph is generated by using the BarabasiAlbertGenerator~\cite{barnard1994fast,batagelj2005efficient} of NetworKit~\cite{staudt2014networkit}. Afterwards, the adjacency matrix of the new graph is interpreted by using the row-net model. 

In \cref{fig:big-bench-1,fig:big-bench-2,fig:soc-1,fig:soc-2} the results are presented graphically.
Each curve plots the ratio 
\[
\frac{\text{cut obtained using another algorithm}}{\text{cut obtained using \algname}}, 
\]
where for the \zoltancolor curve, the ``other algorithm'' is Zoltan; \patohqcolor curve, PaToH-Q; \patohdcolor curve, PaToH-D; and \hmetiscolor curve, hMetis2. Each plot corresponds to a certain number of parts and a certain imbalance factor. The hypergraphs are ordered in the the increasing ratio.

For readability, the results for big-bench are split in two parts: those with ratios
$$\frac{\text{cut obtained using another algorithm}}{\text{cut obtained using \algname}}<1.5$$
and those with ratios greater than $1.5$ (i.e., with more than 50\% improvement in cut). The results for the social-networks-bench and SNAP-bench are plotted on the same figure, since the social-networks-bench is considerably smaller than the other two. The results show substantial improvement over Zoltan without algebraic distance, as well as over hMetis2 and PaToH on most of the hypergraphs. For the full set of results, please refer to \url{http://shaydul.in/hypergraph-partitioning-archive/}.

\begin{figure} 
\captionsetup{justification=centering}
\centering
  \includegraphics[width=0.85\textwidth]{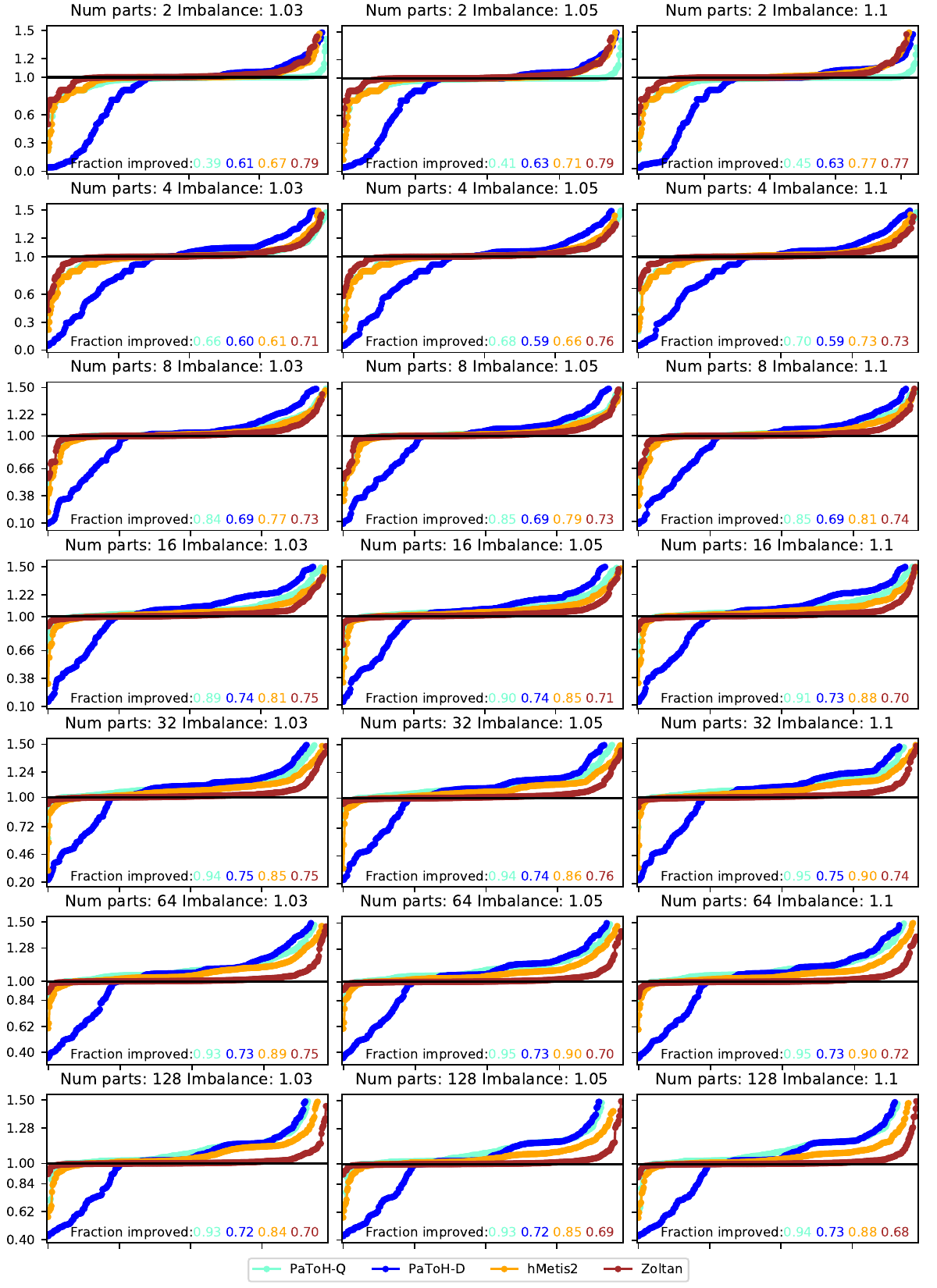}
  \caption{First half ($\frac{\text{cut obtained using another algorithm}}{\text{cut obtained using \algname}}<1.5$) of big-bench (443 hypergraphs), with \zoltancolor curve corresponding to Zoltan, \patohqcolor to PaToH-Q, \patohdcolor to PaToH-D, and \hmetiscolor to hMetis2. The improvements are with respect to the whole big-bench, including those with improvement greater than $1.5$.}
  \label{fig:big-bench-1}
\end{figure}

\begin{figure} 
\captionsetup{justification=centering}
\centering
  \includegraphics[width=0.85\textwidth]{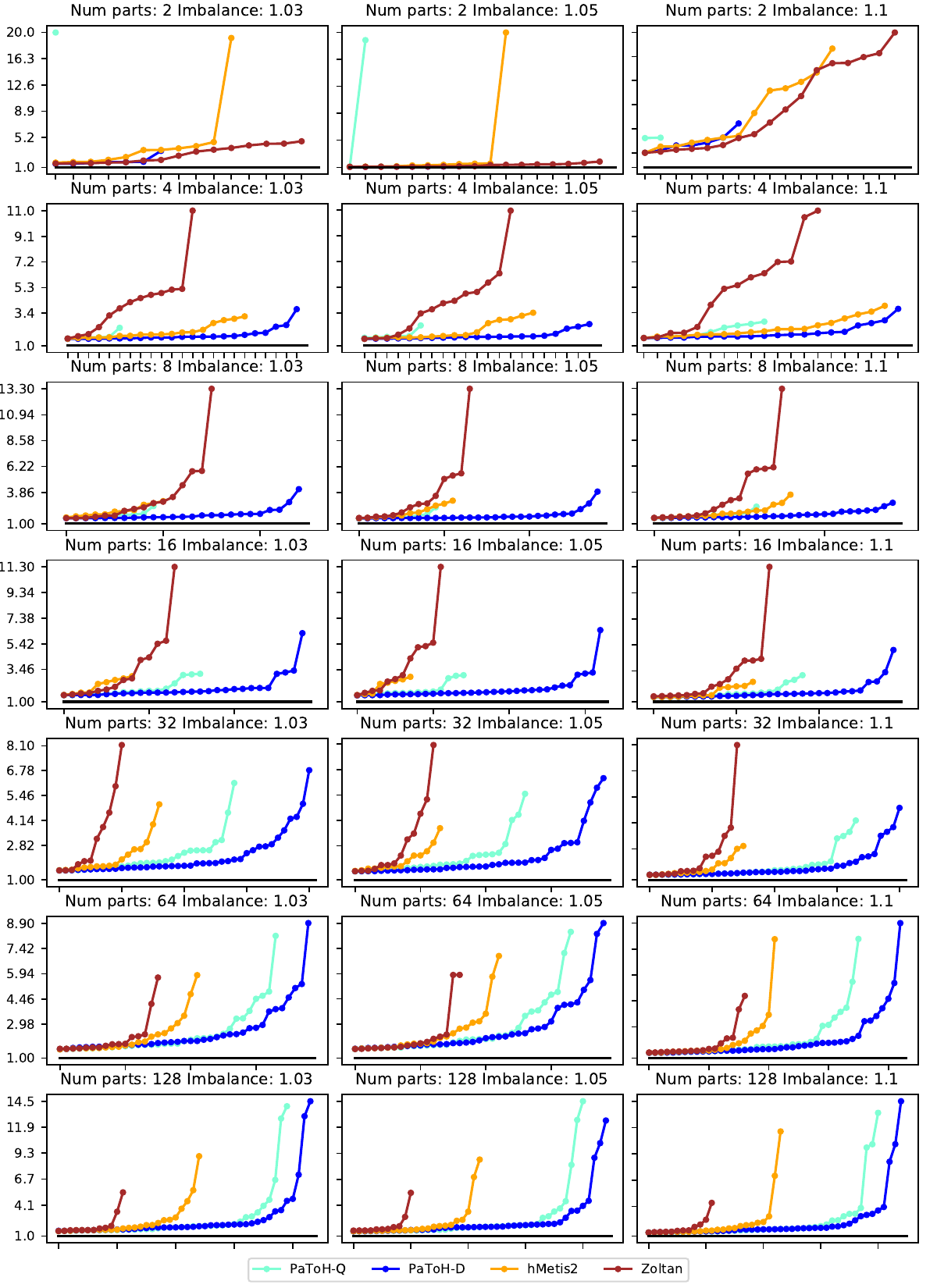}
  \caption{Second half ($\frac{\text{cut obtained using another algorithm}}{\text{cut obtained using \algname}}\geq 1.5$) of big-bench (443 hypergraphs), with \zoltancolor curve corresponding to Zoltan, \patohqcolor to PaToH-Q, \patohdcolor to PaToH-D, and \hmetiscolor to hMetis2 (greater is better).}
  \label{fig:big-bench-2}
\end{figure}

\begin{figure}  
\captionsetup{justification=centering}
\centering
  \includegraphics[width=0.85\textwidth]{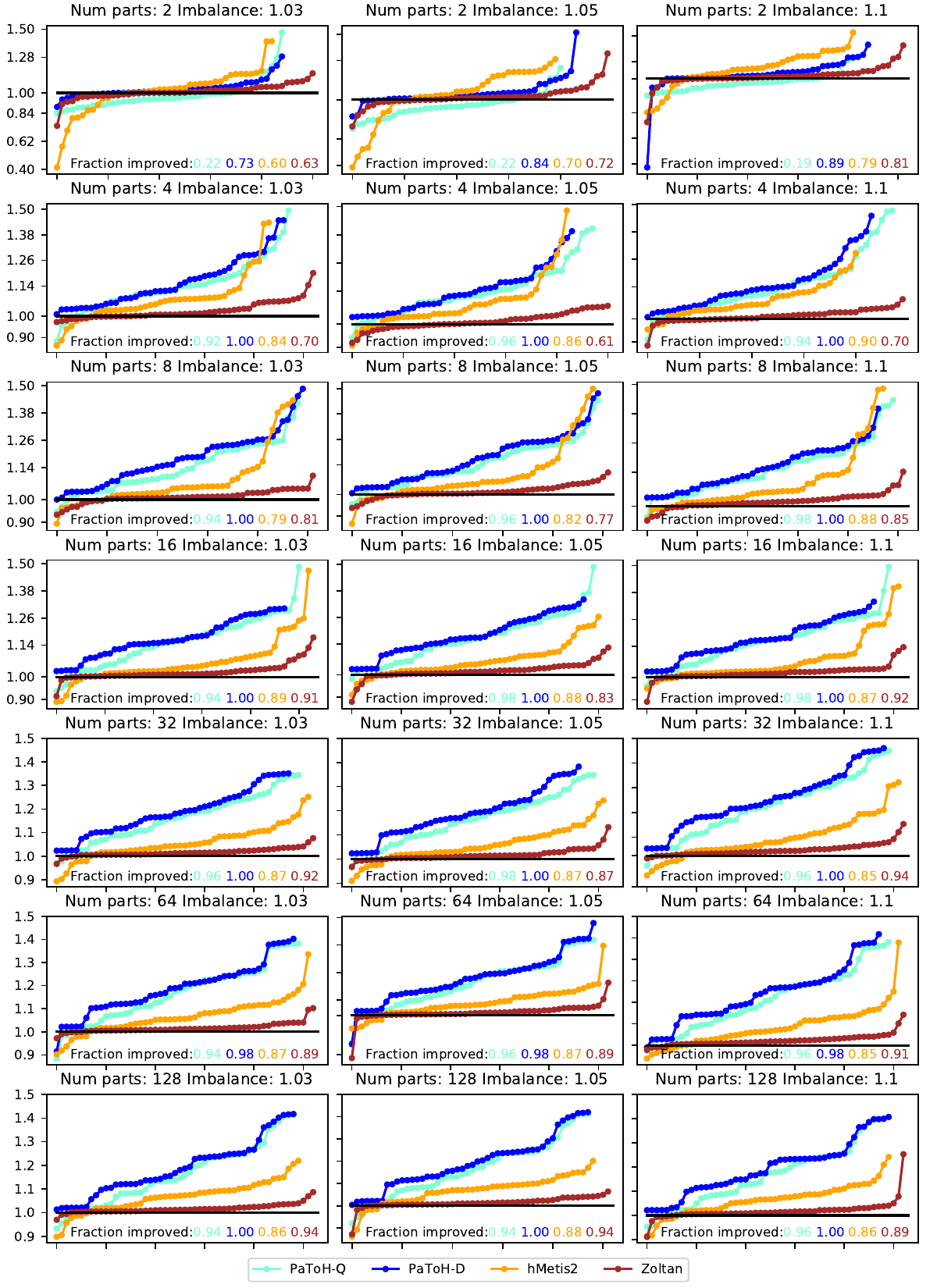}
  \caption{First half ($\frac{\text{cut obtained using another algorithm}}{\text{cut obtained using \algname}}<1.5$) of SNAP-bench and social-networks-bench (54 hypergraphs), with \zoltancolor curve corresponding to Zoltan, \patohqcolor to PaToH-Q, \patohdcolor to PaToH-D, and \hmetiscolor to hMetis2. The improvements are with respect to the whole big-bench, including those with improvement greater than $1.5$}
  \label{fig:soc-1}
\end{figure}

\begin{figure} 
\captionsetup{justification=centering}
\centering
  \includegraphics[width=0.85\textwidth]{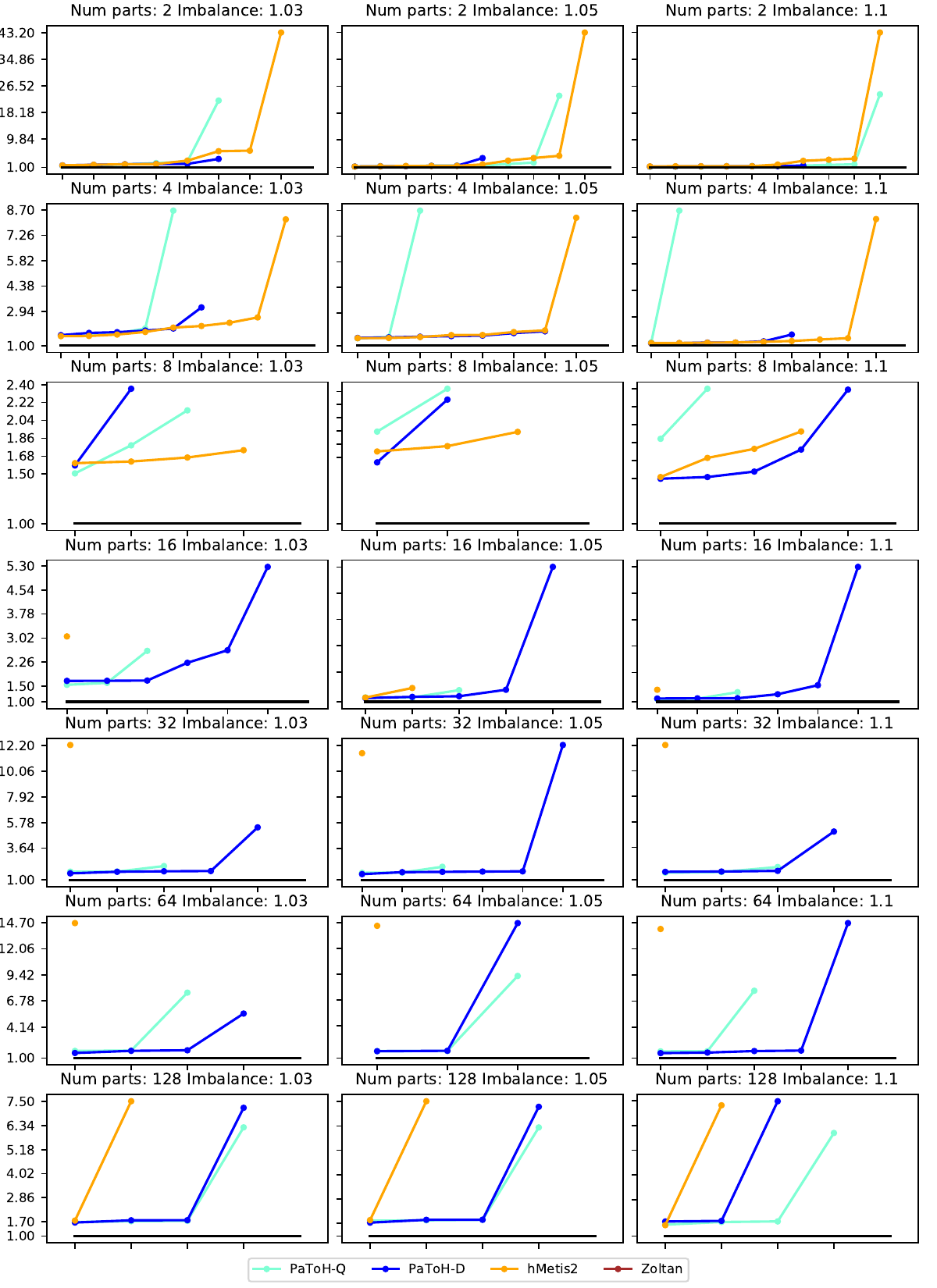}
  \caption{Second half ($\frac{\text{cut obtained using another algorithm}}{\text{cut obtained using \algname}}\geq 1.5$) of SNAP-bench and social-networks-bench (54 hypergraphs), with \zoltancolor curve corresponding to Zoltan, \patohqcolor to PaToH-Q, \patohdcolor to PaToH-D, and \hmetiscolor to hMetis2 (greater is better).}
  \label{fig:soc-2}
\end{figure}

\begin{figure} 
\captionsetup{justification=centering}
\centering
  \includegraphics[width=0.85\textwidth]{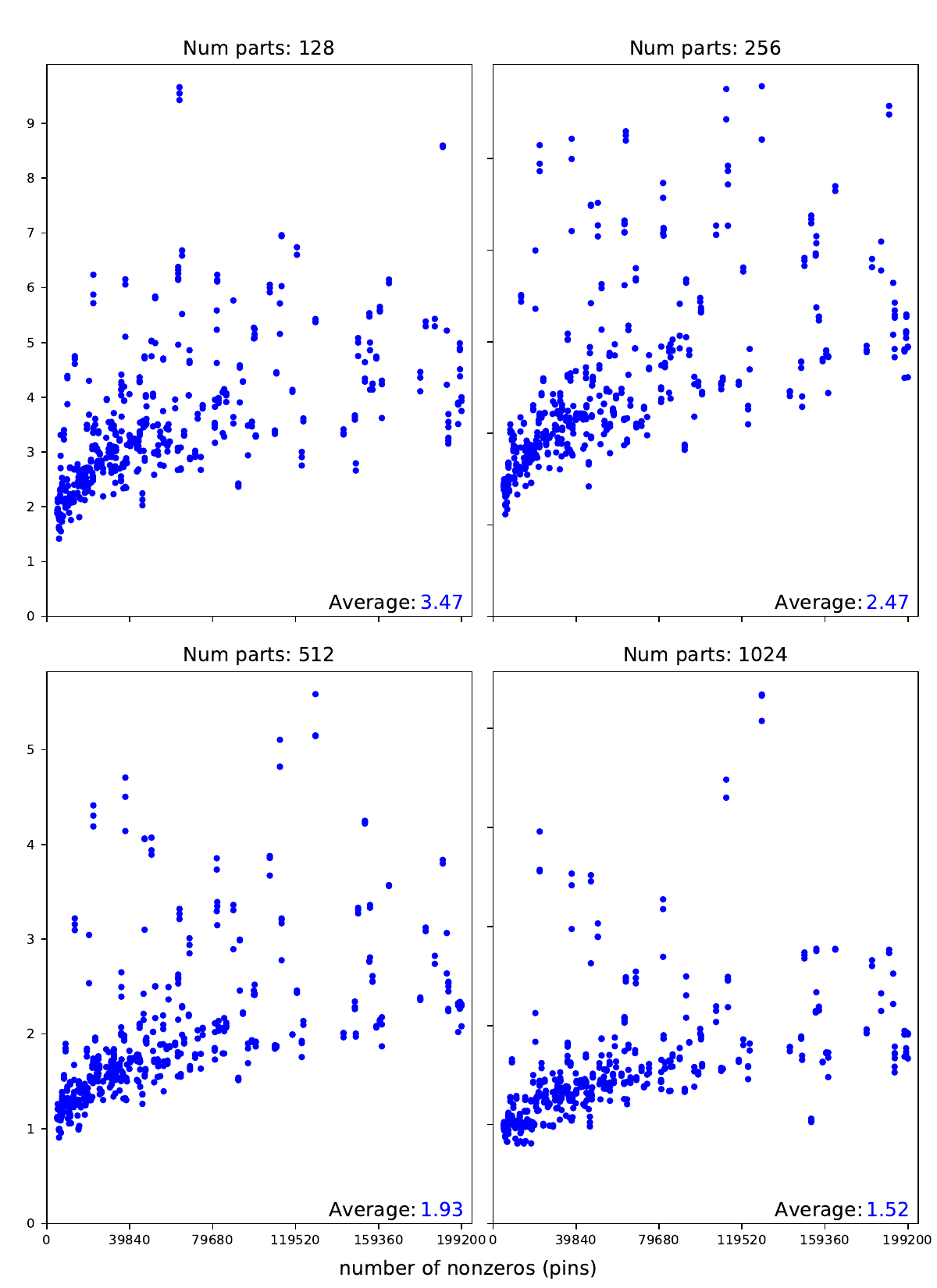}
  \caption{Ratio $\frac{\text{\algname runtime}}{\text{Zoltan runtime}}$ for the big-bench benchmark. Time is CPU time as reported by $zdrive$.}
  \label{fig:time}
\end{figure}
%\medskip

\section{Conclusions}

In this paper we have introduced a new similarity measure for hypergraph vertices---algebraic distances. This similarity measure is used for more accurate vertex aggregation during the coarsening stage of a multilevel algorithm. A~serial iterative procedure for computing algebraic distances is introduced and implemented within the multilevel hypergraph partitioning framework Zoltan. The procedure results in a significant improvement (average of $34.3\%$) over the same framework without algebraic distances, while decreasing the cut by more than two times for some hypergraphs. The algorithm is shown to outperform other state-of-the-art partitioners as well.

The experimental results indicate that one may gain substantial performance improvements through exploiting the global structure of highly irregular hypergraphs (e.g. social networks and other hypergraphs with power-law degree distribution). Exploiting the spectral properties of the hypergraph and its star expansion through some iterative procedure, like the one proposed in the work, is one way to achieve the gain. 
There remains ample room for improvement for current state-of-the-art hypergraph partitioners, particularly for the coarsening stage.

\section*{Acknowledgments}
This material is based upon work supported by the National Science Foundation under grant no. 1522751. The authors would like to thank Clemson University for generous allotment of computation time on Palmetto cluster and Palmetto support staff for technical assistance.

\bibliographystyle{siam}
\bibliography{main}
\end{document}